\documentclass[letterpaper, 10 pt, conference]{ieeeconf} 

\IEEEoverridecommandlockouts          
\usepackage{enumerate}
\usepackage{amsmath}
\usepackage{amssymb}
\usepackage{hyperref}
\usepackage{graphicx}
\usepackage{epstopdf} 
\usepackage{mathtools}
\usepackage{booktabs}
\usepackage[english]{babel}
\usepackage{xcolor, soul}
\usepackage{siunitx}    
\usepackage{lipsum}
\usepackage{mathabx}
\usepackage{accents}
\usepackage{comment}

\usepackage[firstpageonly=true]{draftwatermark}

\newtheorem{assumption}{Assumption}
\newtheorem{definition}{Definition}
\newtheorem{proposition}{Proposition}

\newtheorem{theorem}{Theorem}

\newcommand\scalemath[2]{\scalebox{#1}{\mbox{\ensuremath{\displaystyle #2}}}}

\newcommand\ubar[1]{\underaccent{\bar}{#1}}

\hypersetup{
	colorlinks,
	linkcolor={red!50!black},
	citecolor={blue!50!black},
	urlcolor={blue!80!black}
}

\title{Deep Long-Short Term Memory networks: Stability properties and Experimental validation}

\author{Fabio Bonassi$^{1, *}$ \thanks{$^*$ Corresponding author}, Alessio La Bella$^{1}$, Giulio Panzani$^{1}$, Marcello Farina$^{1}$, Riccardo Scattolini$^{1}$
	\thanks{$^{1}$ The authors are with the Dipartimento di Elettronica, Informazione e Bioingegneria, Politecnico di Milano, Via Ponzio 34/5, 20133, Milano, Italy. E-mail: {\tt\small name.surname@polimi.it}}}

\begin{document}

\maketitle
\thispagestyle{empty}
\pagestyle{empty}

\begin{abstract}
The aim of this work is to investigate the use of Incrementally Input-to-State Stable ($\delta$ISS) deep Long Short Term Memory networks (LSTMs) for the identification of nonlinear dynamical systems.
We show that suitable sufficient conditions on the weights of the network can be leveraged to setup a training procedure able to learn provenly-$\delta$ISS LSTM models from data.
The proposed approach is tested on a real brake-by-wire apparatus to identify a model of the system from input-output experimentally collected data.
Results show satisfactory modeling performances.
\end{abstract}


 \DraftwatermarkOptions{%
 angle=0,
 hpos=0.5\paperwidth,
 vpos=0.96\paperheight,
 fontsize=0.012\paperwidth,
 color={[gray]{0.2}},
 text={
   \parbox{0.99\textwidth}{This manuscript is an extended version of a paper accepted for the 2023 European Control Conference (ECC'23). Copyright may be transferred without notice, after which this version may no longer be accessible.}},
 }

\section{Introduction}

In recent years, a progressive interchange of ideas between the machine learning and the control systems communities has led to a fruitful diffusion of machine learning techniques and tools in the context of identification and control of dynamical systems.
We refer, in particular, to Neural Networks (NNs): albeit their use in the context of the control system design has been investigated since the nineties \cite{hunt1992neural}, the availability of large and informative datasets, of cheap computational power, and of open-source software tools for implementing and training these networks has renewed the interest in these techniques \cite{goodfellow2016deep}.

Among the many NN architectures available, Recurrent Neural Networks (RNNs) are known to be the most suitable for system identification purposes \cite{bianchi2017recurrent}, since their stateful nature allows them to store memory of past data.
Particular successful and popolar have been Long Short-Term Memory networks (LSTMs, \cite{hochreiter1997long}) and Gated Recurrent Units (GRUs, \cite{cho2014learning}), whose gated architectures make them less prone to the well-known vanishing and exploding gradient problems which normally plague RNNs \cite{pascanu2013difficulty}.
In contrast, these two architectures have demonstrated remarkable capabilities in learning dynamical systems \cite{mohajerin2019multistep}.

In the context of control NNs can be employed in a vast multitude of ways
\cite{xu2022review, bonassi2022survey}, including indirect data-driven control synthesis.
In particular, NN models of dynamical systems, identified from input-output data collected through experiments \cite{schoukens2019nonlinear}, can then be used to synthesize model-based controllers, e.g., with Model Predictive Control (MPC) \cite{rawlings2017model}.

This strategy, which allows to naturally blend machine-learning methods (for system identification) with traditional model-based control design strategies, has gained significant interest in both academia, see e.g. \cite{wu2019machine, levin1996control, terzi2021lstm} and in industry, see e.g. \cite{lanzetti2019recurrent, nagy2007model, wong2018recurrent}.
Once the gated RNN model of the system is identified, in accordance with the indirect data-driven control paradigm, the Certainty Equivalence Principle (CEP) is invoked, that allows the controller to be designed based on the identified model, e.g. by resorting to MPC strategies.
In this context, being able to certify the model's Incremental Input-to-State Stability ($\delta$ISS, \cite{bayer2013discrete}) allows the designer to synthesize control laws with closed-loop stability guarantees \cite{terzi2021lstm, bonassi2022imc} and robust zero-error tracking of constant reference signals \cite{bonassi2021nnarx, bonassi2022offset}.
As discussed in \cite{bonassi2022survey}, this property is also beneficial in terms of model's robustness against input perturbations, consistency of modeling performances in spite of wrong initialization, and safety \cite{bonassi2020lstm}.

The problem of training provenly-$\delta$ISS RNNs has been recently considered in the literature.
Sufficient conditions for the networks' $\delta$ISS have been devised in \cite{bonassi2021nnarx} for Neural NARXs networks, i.e., ARX models with feed-forward NNs as nonlinear regressors.
In \cite{bonassi2021stability}, the $\delta$ISS of GRUs has been investigated, devising conditions for the $\delta$ISS of both \emph{shallow} (i.e., single-layer) and \emph{deep} (i.e., multi-layer) GRUs.
Sufficient conditions for the $\delta$ISS of LSTMs have been provided in \cite{terzi2021lstm}, but limitedly to shallow networks, which may not be suited to identify complex real systems \cite{bianchi2017recurrent}.
\smallskip

The objective of this paper is twofold: (\emph{i}) to establish $\delta$ISS conditions for deep LSTM networks, (\emph{ii}) to apply the proposed framework for the identification of a real system with nontrivial dynamics, i.e., a brake-by-wire apparatus.
\smallskip

The paper is structured as follows. 
In Section \ref{sec:lstm}, the state-space formulation of deep LSTM networks is described, and the sufficient conditions for their $\delta$ISS are devised in Section \ref{sec:stability}.
In Section \ref{sec:example}, the proposed approach is tested on a real brake-by-wire apparatus, allowing to identify a $\delta$ISS model of the system.

\begin{figure*}[t]
	\centering
	\includegraphics[width=0.85\textwidth]{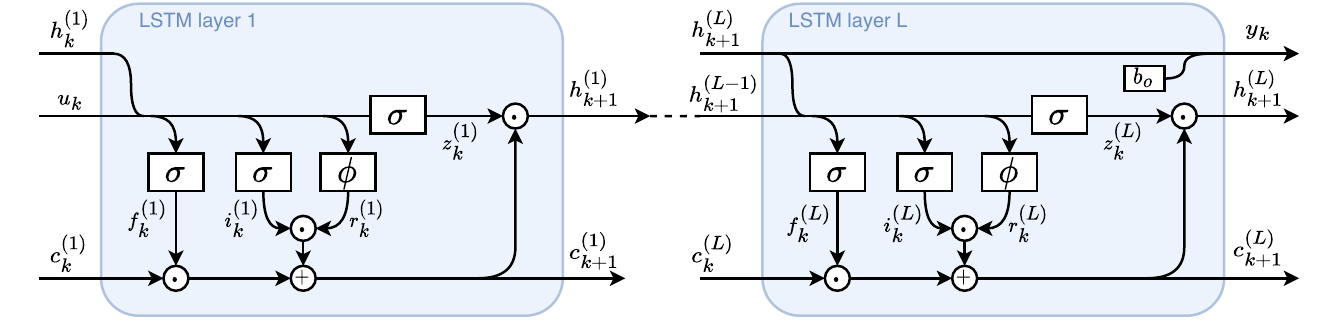}
	\caption{Schematic of a deep LSTM architecture, consisting in the concatenation of $L$ LSTM layers, where each layer takes the updated hidden state of the preceding layer as input.}
	\label{fig:deeplstm}
\end{figure*}

\subsection*{Notation}
The following notation is adopted in the paper. 
Given a vector $v \in \mathbb{R}^n$, we denote by $v^\prime$ its transpose, by $[ v ]_i$ its $i$-th component, and by $\| v \|_p$ its $p$-norm.
The Hadamard (element-wise) product between two vectors $v$ and $w$ of the same dimensions is denoted as $v \circ w$.
Given a $n$-by-$m$ matrix $A$, $ \| A \|_p$ is used to indicate its induced $p$-norm.

For the sake of conciseness, for time-varying vectors the time index $k$ is reported as a subscript, e.g. $v_k$. 
Sequences of vectors spanning from the time index $k_1$ to $k_2 \geq k_1$ are denoted as $v_{k_1 : k_2}$, i.e. $v_{k_1 : k_2} = \{ v_{k_1}, v_{k_1 +1}, ..., v_{k_2} \}$.
The $\ell_{p, q}$ norm of a sequence is defined as $\| v_{k_1 : k_2} \|_{p, q} = \big\| [ \| v_{k_1} \|_p,  ..., \| v_{k_2} \|_p ]^\prime \big\|_q$,
for which a notable case is the $\ell_{p, \infty}$ norm, where $\| v_{k_1 : k_2} \|_{p, \infty} = \max_{k \in \{ k_1, ..., k_2\} } \| v_k \|_p$. 

Finally, we denote by $\sigma(x) = \frac{1}{1 + e^{-x}}$  and by $\phi(x) = \tanh(x)$ the sigmoidal  and $\tanh$ activation functions, respectively.
Note that when applied to vectors, these activation functions are intended to be applied element-wise.

\section{Deep LSTM architecture} \label{sec:lstm}
Motivated by the well-known superior modeling power of deep recurrent networks \cite{bianchi2017recurrent}, in this paper we focus on deep LSTMs for nonlinear system identification.

\begin{subequations} \label{eq:lstm:model}
Deep LSTMs consist of the concatenation of $L > 1$ LSTM layers, where each layer is described by the state-space equation of shallow LSTMs, see Figure \ref{fig:deeplstm}.
In particular, each layer $l \in \{ 1, ..., L \}$ is ruled by
\begin{equation} \label{eq:lstm:model:statespace}
\begin{dcases}
	c_{k+1}^{(l)} = f_k^{(l)} \circ c_{k}^{(l)} + i_{k}^{(l)} \circ r_k^{(l)}, \\
	h_{k+1}^{(l)} = z_{k}^{(l)} \circ \phi(h_{k}^{(l)}),
\end{dcases}
\end{equation}
where $c_k^{(l)} \in \mathbb{R}^{n_c^{(l)}}$ and $h_k^{(l)} \in \mathbb{R}^{n_c^{(l)}}$ are the cell and hidden states, respectively, with $n_{c}^{(l)}$ indicating the (a priory selected) amount of neurons of the layer, chosen.
The layer's state vector is the concatenation of the cell and hidden states, i.e., $x_k^{(l)} = [ c_k^{(l) \prime}, h_{k}^{(l) \prime} ]^\prime \in \mathbb{R}^{n_x^{(l)}}$, with $n_x^{(l)} = 2 n_c^{(l)}$.

The terms $f_k^{(l)}$, $i_k^{(l)}$, and $z_k^{(l)}$ appearing in \eqref{eq:lstm:model:statespace} are the so-called \emph{gates} of the network, which allow to regulate the flow and storage of information within the layer, while $r_k^{(l)}$ is the so-called \emph{squashed input}. They read as
\begin{equation} \label{eq:lstm:model:gates}
	\begin{aligned}
		f_k^{(l)} &= \sigma \big( W_f^{(l)} u_k^{(l)} + U_f^{(l)} h_k^{(l)}  + b_f^{(l)} \big), \\
		i_k^{(l)} &= \sigma \big( W_i^{(l)} u_k^{(l)} + U_i^{(l)} h_k^{(l)}  + b_i^{(l)} \big), \\
		z_k^{(l)} &= \sigma \big( W_z^{(l)} u_k^{(l)} + U_z^{(l)} h_k^{(l)}  + b_z^{(l)}\big), \\
		r_k^{(l)} &= \phi \big( W_r^{(l)} u_k^{(l)} + U_r^{(l)} h_k^{(l)}  + b_r^{(l)}\big).
	\end{aligned}
\end{equation}
Note that, since the sigmoid function $\sigma(\cdot)$ is bounded in $(0, 1)$, the gates are vectors of subunitary positive elements: gates whose components are close to $1$ are said to be open, i.e., gates allowing the flow of information; those whose components are close to $0$ are said to be closed.
Similarly, since the $\tanh$ function $\phi(\cdot)$ is bounded in $(-1, 1)$, the squashed input is a vector of components lying in such range.

The gates and squashed input introduced in \eqref{eq:lstm:model:gates} also depend on $u_k^{(l)}$, i.e., the layer's input vector.
In deep architectures, such input is indeed defined as
\begin{equation} \label{eq:lstm:model:input}
	u_k^{(l)} =  \left\{ \begin{array}{ll}
		u_{k} & \text{if } \, l = 1, \\
		h_{k+1}^{(l-1)} \quad & \text{if } \, l \in \{ 2, ..., L \}.
	\end{array} \right.
\end{equation}
Specifically, the first layer is fed by the model's input $u_k$, while the input of any following layer $l \in \{ 2, ..., L \}$ is defined as the updated hidden state of the immediately preceding layer, i.e., $h_{k+1}^{(l-1)}$.
Finally, the model's output is defined as an affine transformation of the last layer's hidden state, 
\begin{equation}
	y_k = U_o h_k^{(L)} + b_o.
\end{equation}
\end{subequations}

Note that each layer is parametrized by the set of weights
\begin{equation*}
\begin{aligned}
	\Phi^{(l)} = \{ &W_{f}^{(l)}, U_f^{(l)}, b_f^{(l)}, W_{i}^{(l)}, U_i^{(l)}, b_i^{(l)}, W_{z}^{(l)}, U_z^{(l)}, b_z^{(l)}, \\
	&W_{r}^{(l)}, U_r^{(l)}, b_r^{(l)} \}.
\end{aligned}
\end{equation*}
Thus, letting  $\Phi = \{ U_o, b_o \} \cup \bigcup_{l=1}^L \Phi^{(l)} $ be the set of network's weights and denoting by $x_k = [ x_k^{(1) \prime}, ..., x_k^{(L) \prime} ]^\prime$
the state vector of the deep LSTM model, \eqref{eq:lstm:model} can be compactly represented by
\begin{equation} \label{eq:lstm:compact}
	\Sigma(\Phi): \, \begin{dcases}
		x_{k+1} = \varphi(x_{k}, u_k; \Phi) \\
		y_k = \psi(x_k; \Phi)
	\end{dcases},
\end{equation}
where $\varphi$ and $\psi$ are suitable function arising from equations \eqref{eq:lstm:model}.
We now introduce a customary assumption concerning the boundedness of the model's input. 
\begin{assumption} \label{ass:input}
	The input vector $u_k$ is unity bounded, i.e., $u_k \in \mathcal{U}$, where 
	\begin{equation}
		\mathcal{U} = \{ u \in \mathbb{R}^{n_u} : \| u \|_\infty \leq 1 \}.
	\end{equation}
\end{assumption}

Note that, when the control variables are subject to saturation, Assumption \ref{ass:input} can be easily fulfilled by means of a suitable normalization of the input, which is a common practice in deep learning \cite{goodfellow2016deep}.
The following proposition can hence be stated.
\smallskip

\begin{proposition} \label{prop:gates_bounds}
	For any layer $l \in \{ 1, ..., L \}$, the squashed input $r_k^{(l)}$ and the gates $f_k^{(l)}$, $i_k^{(l)}$, and $z_k^{(l)}$ are bounded as 
	\begin{subequations}
		\begin{align}
			-1 < - \bar{\phi}^{(l)}_r &\leq [ r_k^{(l)} ]_j \leq \bar{\phi}_r^{(l)} < 1, \\
			0 < 1 - \bar{\sigma}_f^{(l)} & \leq [ f_k^{(l)} ]_j \leq \bar{\sigma}_f^{(l)} < 1,  \\
			0 < 1 - \bar{\sigma}_i^{(l)} & \leq [ i_k^{(l)} ]_j \leq \bar{\sigma}_i^{(l)} < 1, \\
			0 < 1 - \bar{\sigma}_z^{(l)} & \leq [ z_k^{(l)} ]_j \leq \bar{\sigma}_z^{(l)} < 1,
 		\end{align}
 		$\forall j \in \{ 1, ..., n_c^{(l)} \}$, where the bounds are defined as
	\end{subequations}
	\begin{subequations} \label{eq:lstm:gates_bounds}
		\begin{align}
			\bar{\phi}^{(l)}_r &= \phi( \| W_r^{(l)} \quad U_r^{(l)} \quad b_r^{(l)} \|_\infty), \\
			\bar{\sigma}^{(l)}_f &= \sigma( \| W_f^{(l)} \quad U_f^{(l)} \quad b_f^{(l)} \|_\infty), \\
			\bar{\sigma}^{(l)}_i &= \sigma( \| W_i^{(l)} \quad U_i^{(l)} \quad b_i^{(l)} \|_\infty), \\
			\bar{\sigma}^{(l)}_z &= \sigma( \| W_z^{(l)} \quad U_z^{(l)} \quad b_z^{(l)} \|_\infty).
 		\end{align}
	\end{subequations}
\end{proposition}	
\begin{proof}
	See Appendix \ref{proof:gates_bounds}.
\end{proof}
\smallskip

In light of these bounds, an invariant set of \eqref{eq:lstm:model} can be defined.
In the next section, the $\delta$ISS property will then be analyzed with respect to this invariant set. 
We recall that a set $\mathcal{X}$ is said to be an invariant set if, for any $u \in \mathcal{U}$, $x \in \mathcal{X} \implies \varphi(x, u) \in \mathcal{X}$, with $\varphi(\cdot)$ defined in \eqref{eq:lstm:compact}. 
\smallskip

\begin{proposition} \label{prop:invariant_set}
	\begin{subequations} \label{eq:lstm:invset}
	An invariant set of the deep LSTM model \eqref{eq:lstm:model} is
	\begin{equation} \label{eq:lstm:invset:full}
		\mathcal{X} = \bigtimes_{l=1}^L \mathcal{X}^{(l)},
	\end{equation}
	where $\bigtimes_{l=1}^L $ indicates the Cartesian product and $\mathcal{X}^{(l)}$ is
	\begin{align}
		\mathcal{X}^{(l)} &= \mathcal{C}^{(l)} \times \mathcal{H}^{(l)}, \text{ where}\\
		\mathcal{C}^{(l)} &=  \bigg\{ c \in \mathbb{R}^{n_c^{(l)}} :  \| c \|_\infty \leq \frac{\bar{\sigma}^{(l)}_i \, \bar{\phi}^{(l)}_r}{1 - \bar{\sigma}^{(l)}_f} = \bar{c}^{(l)}   \bigg\}, \\
		\mathcal{H}^{(l)} &= \bigg\{ h \in \mathbb{R}^{n_c^{(l)}} : \| h \|_\infty \leq \phi(\bar{c}^{(l)}) = \bar{h}^{(l)} \bigg\},
	\end{align}
	and $\bar{\sigma}^{(l)}_f$, $\bar{\sigma}^{(l)}_i$, $\bar{\sigma}^{(l)}_z$, and $\bar{\sigma}^{(l)}_r$  are defined as in \eqref{eq:lstm:gates_bounds}.
	\end{subequations}
\end{proposition}
\begin{proof}
	See Appendix \ref{proof:invariant_set}.
\end{proof}
\smallskip

We are now in the position to analyze the $\delta$ISS property of the deep LSTM model \eqref{eq:lstm:compact}.

\section{Stability properties} \label{sec:stability}
\begin{definition}[$\mathcal{K}_\infty$ function]
	A continuous function $\Psi(s) : \mathbb{R}_{\geq 0} \mapsto \mathbb{R}_{\geq 0}$ is of class $\mathcal{K}_\infty$ if it is strictly increasing with its argument, $\Psi(0) = 0$, and $\Psi(s) \xrightarrow[s \to \infty]{} \infty$.
\end{definition}
\vspace{1mm}

\begin{definition}[$\mathcal{KL}$ function]
	A continuous function $\Psi(s) : \mathbb{R}_{\geq 0} \times \mathbb{R}_{\geq 0} \mapsto \mathbb{R}_{\geq 0}$ is of class $\mathcal{KL}$ if it is of class $\mathcal{K}_\infty$ with respect to its first argument and, for any $s$, $\Psi(s, t) \xrightarrow[t \to \infty]{} 0$.
\end{definition}

Let us denote by $x_{0:k}(x_0, u_{0:k})$ the state trajectory of \eqref{eq:lstm:compact} when initialized in $x_{0}$ and fed with the input sequence $u_{0:k}$.
The (regional) $\delta$ISS considered in this paper can be hence defined as follows.
\smallskip

\begin{definition}[$\delta$ISS \cite{bayer2013discrete}]
	System \eqref{eq:lstm:compact} is Incrementally Input-to-State Stable ($\delta$ISS) with respect to the input set $\mathcal{U}$ and to the invariant set $\mathcal{X}$ if there exist functions $\beta$ of class $\mathcal{KL}$ and $\gamma$ of class $\mathcal{K}_\infty$ such that, for any pair of initial states $x_{a, 0} \in \mathcal{X}$ and $x_{b, 0} \in \mathcal{X}$, and any pair of input sequences $u_{a, 0:k} \in \mathcal{U}$ and $u_{b, 0:k} \in \mathcal{U}$, it holds that
	\begin{equation}
		\| x_{a, k} - x_{b, k} \|_2 \leq \beta(\| x_{a, 0} - x_{b, 0} \|_2, k) + \gamma(\| u_{a, 0:k} - u_{b, 0:k} \|_{2, \infty}), 
	\end{equation}
	where $x_{\alpha, k}$ is short for $x_{k}(x_{\alpha, 0}, u_{\alpha, 0:k})$. 
\end{definition}
\smallskip

Remarkably, this property implies that by initializing system \eqref{eq:lstm:compact} in two different initial states within $\mathcal{X}$, and feeding it with any two input sequences extracted from the set $\mathcal{U}$, one obtains state trajectories whose distance is asymptotically independent of the initial conditions.
Moreover, the closer are such input sequences, the smaller is the bound on the distance between the state trajectories.

The $\delta$ISS property is a useful tool to guarantee the safety and robustness of RNNs against input perturbations \cite{bonassi2022survey}.
Indeed, this property can be leveraged to compute a conservative bound of the model's output reachable set or, at least, to provide a theoretical ground for the numerical computation of such set.
In \cite{terzi2021lstm, bonassi2021nonlinear, bonassi2022offset}, the model's $\delta$ISS has also been shown to allow for the synthesis of nominally closed-loop stable predictive control laws.

The goal of this paper is to analyze the $\delta$ISS properties of deep LSTMs, so as to pave the way for the design of predictive control laws based on these models guaranteeing stability of the closed-loop system. 
In particular, we here show that Propositions \ref{prop:gates_bounds} and \ref{prop:invariant_set} imply that, if the $\delta$ISS sufficient condition devised in Proposition 2 of \cite{terzi2021lstm} is satisfied layer-wise, then the deep LSTM is $\delta$ISS.
\smallskip

\begin{theorem} \label{theorem:deltaiss}
	A sufficient condition for the $\delta$ISS of the deep LSTM \eqref{eq:lstm:compact} with respect to the sets $\mathcal{X}$ and $\mathcal{U}$ is that, for each layer $l \in \{ 1, ..., L \}$, 
	\begin{subequations} \label{eq:stability:conditions}
	\begin{align}
		\scalemath{0.9}{
		\bar{\sigma}_f^{(l)} + \bar{\sigma}_z^{(l)} \bar{\alpha}^{(l)} + \frac{1}{4} \bar{h}^{(l)} \| U_z^{(l)} \|_2 -  \frac{1}{4} \bar{\sigma}_f^{(l)}  \bar{h}^{(l)} \| U_z^{(l)} \|_2 - 1} & \scalemath{0.9}{< 0}, \\
		\scalemath{0.9}{\frac{1}{4} \bar{\sigma}_f^{(l)}  \bar{h}^{(l)} \| U_z^{(l)} \|_2 - 1} &\scalemath{0.9}{< 0},
	\end{align}
	\end{subequations}
	where $\bar{\alpha}^{(l)}$ is defined as
	\begin{equation}
		\bar{\alpha}^{(l)} = \frac{1}{4} \bar{c}^{(l)} \| U_f^{(l)} \|_2 + \bar{\sigma}_i^{(l)} \| U_r^{(l)} \|_2 + \frac{1}{4} \bar{\phi}_r^{(l)} \| U_i^{(l)} \|_2, 
	\end{equation}
	$\bar{\sigma}_f^{(l)}$, $\bar{\sigma}_i^{(l)}$, $\bar{\sigma}_z^{(l)}$, and $\bar{\phi}_r^{(l)}$ are defined as in \eqref{eq:lstm:gates_bounds}, and $\bar{c}^{(l)}$ and $\bar{h}^{(l)}$ as in \eqref{eq:lstm:invset}.
\end{theorem}
\begin{proof}
	See Appendix \ref{proof:deltaiss}.
\end{proof}
\smallskip

Note that the sufficient condition devised in Theorem \ref{theorem:deltaiss} correspond to $2 L$ nonlinear inequalities on the weights of the deep LSTM model.
For the sake of compactness, this set of inequalities is denoted as 
\begin{equation} \label{eq:stability:nu}
	\nu(\Phi) < 0.
\end{equation}
Note that this result is consistent with \cite{angeli2002lyapunov}, where the $\delta$ISS of cascaded continuous-time $\delta$ISS systems has been investigated. 
Condition~\eqref{eq:stability:nu} can be used in two alternative ways: (\emph{i}) assessing a-posteriori whether the $\delta$ISS of a trained model can be guaranteed, (\emph{ii}) enforcing, during the training procedure, the satisfaction of $\nu(\Phi) < 0$, and hence the $\delta$ISS of the network being trained. 
Based on this idea, in Appendix \ref{appendix:training} a procedure for training provenly-$\delta$ISS deep LSTM models is detailed.

\section{Numerical example}\label{sec:example}
\begin{figure}[t]
	\centering
	\includegraphics[width=0.9 \columnwidth]{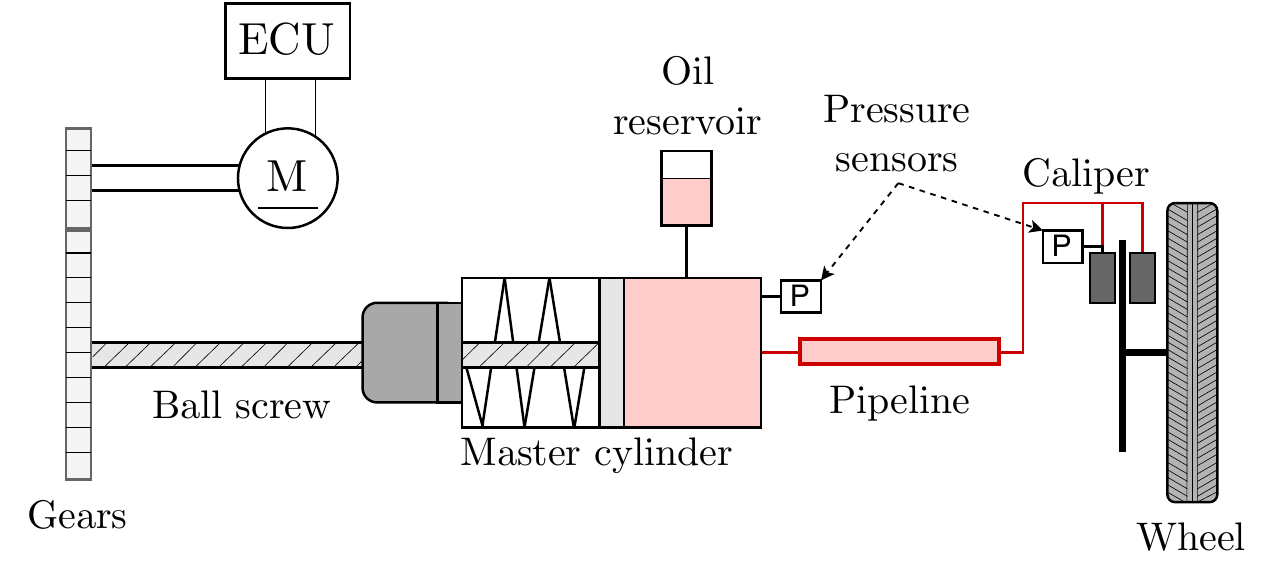}
	\includegraphics[width=0.75 \columnwidth]{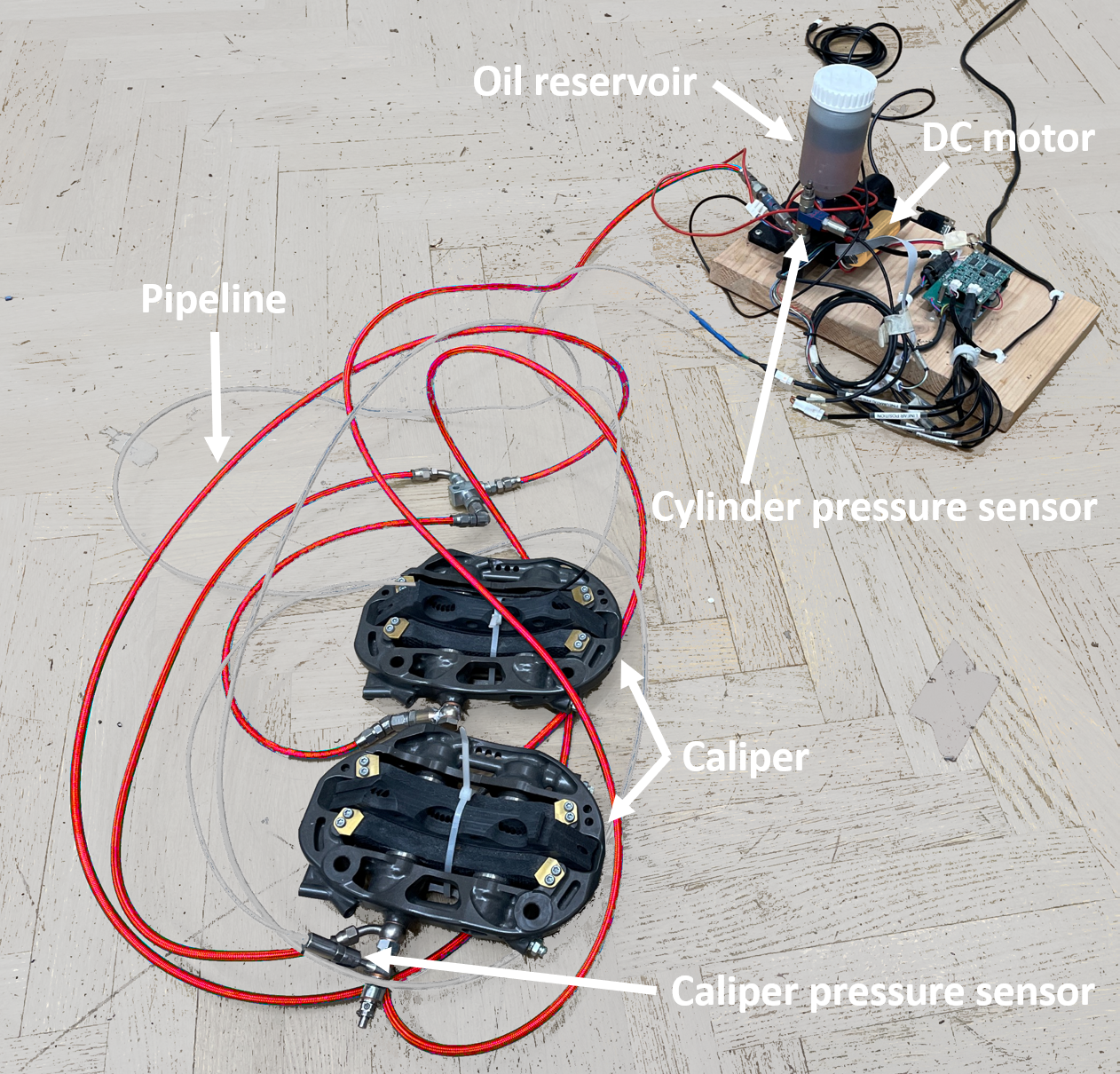}
	\vspace{-2mm}
	\caption{Schematic of the brake-by-wire prototype apparatus.}
	\label{fig:example:bbw}
\end{figure}
The modeling performances of the proposed approach are tested against the data of a brake-by-wire (BBW) actuator.
Brake-by-wire actuators represent an application of the drive-by-wire paradigm to the braking component of a vehicle: by physically decoupling the driver command with respect to the actual braking pressure, they are effectively employed by high-level vehicle system controls, such as Anti-lock Braking Systems (ABS) and Electronic Stability Programs (ESP). 

The brake-by-wire actuator is shown in Figure~\ref{fig:example:bbw}, along with its schematic representation.
The system features a DC motor coupled with a ball screw mechanism, which converts the motor rotary motion into the master cylinder linear one. By changing the master cylinder position the pressure in the braking calipers is built; eventually, pressing the braking pads onto the braking disk, it yields the desired braking torque.
The ultimate goal of a BBW control system is the precise tracking of a reference braking pressure. An effective, widely employed control strategy, features two nested loops \cite{todeschini2014adaptive}: the inner one controls the master cylinder position $x_p$, by acting on the DC motor current $i_{m}$; the outer one controls the circuit braking pressure, by acting on the internal loop position reference. Despite the variable of interest is the caliper pressure $P_{cal}$,  to simplify the overall system layout the master cylinder pressure $P_{cyl}$ is often used as feedback.

The peculiarity of the considered setup is the very long pipeline, highlighted with the color red in Figure~\ref{fig:example:bbw}, which causes non-negligible dynamics differences between the master cylinder and the caliper pressure, \cite{corno2015BBWmodeling}. This motivates, for control design and virtual sensing purposes, the development of a model capable of correctly describing the relationship between the master cylinder position $x_p$, the pressure in the cylinder $P_{cyl}$, and the pressure at the caliper $P_{cal}$. We here propose to learn such model by means of a $\delta$ISS deep LSTM network. The choice of using a non-linear dynamic model like the LSTM one, well fits with the BBW features, i.e. the presence of non-linearities and its Single-Input Multi-Output (SIMO) nature.

\subsection{Experiment campaign and data collection}
To collect the necessary data and complete the model training procedure, $6$ open-loop experiments have been conducted on the apparatus.
In each experiment, lasting $180$ seconds, a multilevel pseudo-random binary signal has been generated for the motor's current reference, which is then actuated by the apparatus' ECU.
Thus, the current reference is a train of steps of random amplitude (in the range $0 \text{A}$ to $10 \text{A}$) and random time duration (between $0.5 \text{s}$ and $1.5 \text{s}$).
For each experiment, these random variables have been drawn from different piece-wise uniform probability distributions. 
In this way, the informativeness of the collected data is enhanced.

During the experiments, $x_p$, $P_{cyl}$, and $P_{cal}$ have been measured with a sampling frequency of $200 \text{Hz}$. 
Each experiment thus consists of $36000$ $(x_p, P_{cyl}, P_{cal})$ datapoints.

The experiments have then been partitioned into training, validation, and test datasets, according to the details provided in Appendix \ref{appendix:example}.

\subsection{Deep LSTM training}
The training procedure has been carried out with PyTorch~1.12 running on Python 3.10.
A deep LSTM with $L = 2$ layers, with $n_c^{(1)} = n_c^{(2)} = 8$ units, has been adopted.
Note that this model has one input (the piston position $x_p$) and two outputs (the pressures in the caliper and in the master cylinder, i.e., $P_{cal}$ and $P_{cyl}$, respectively).

\begin{figure}[t]
	\centering
	\includegraphics[width=\columnwidth]{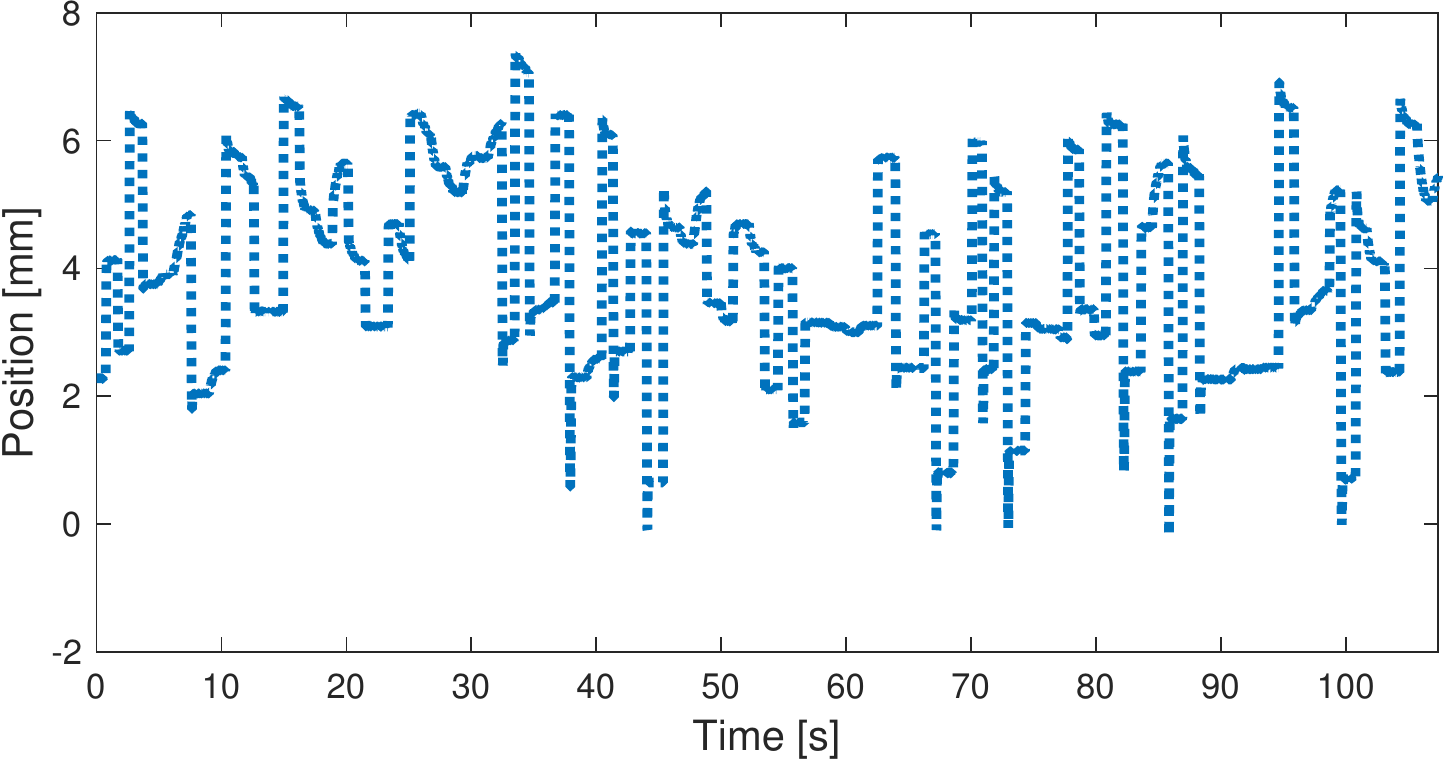} 
	\caption{Performance testing: input sequence (i.e., the piston position) applied to the deep LSTM model.}
	\label{fig:example:test_input}
\end{figure}

The model has been trained by minimizing the simulation MSE while fulfilling conditions \eqref{eq:stability:conditions}, following the procedure described in Appendix \ref{appendix:training}, whereas in Appendix \ref{appendix:example} the implementation of the training procedure is detailed.
The resulting modeling performances can be appreciated in the following Figures. In particular, Figure \ref{fig:example:Y_test_LSTM_vs_linear} shows an overview of the test data along with the simulated trajectories, using the test input of Figure \ref{fig:example:test_input}, of the learned LSTM model and those of a linear model. The linear model has been identified using a classic subspace identification method, that suits with the SIMO nature of the system.
\begin{figure}[t]
	\centering
	\includegraphics[width= \columnwidth]{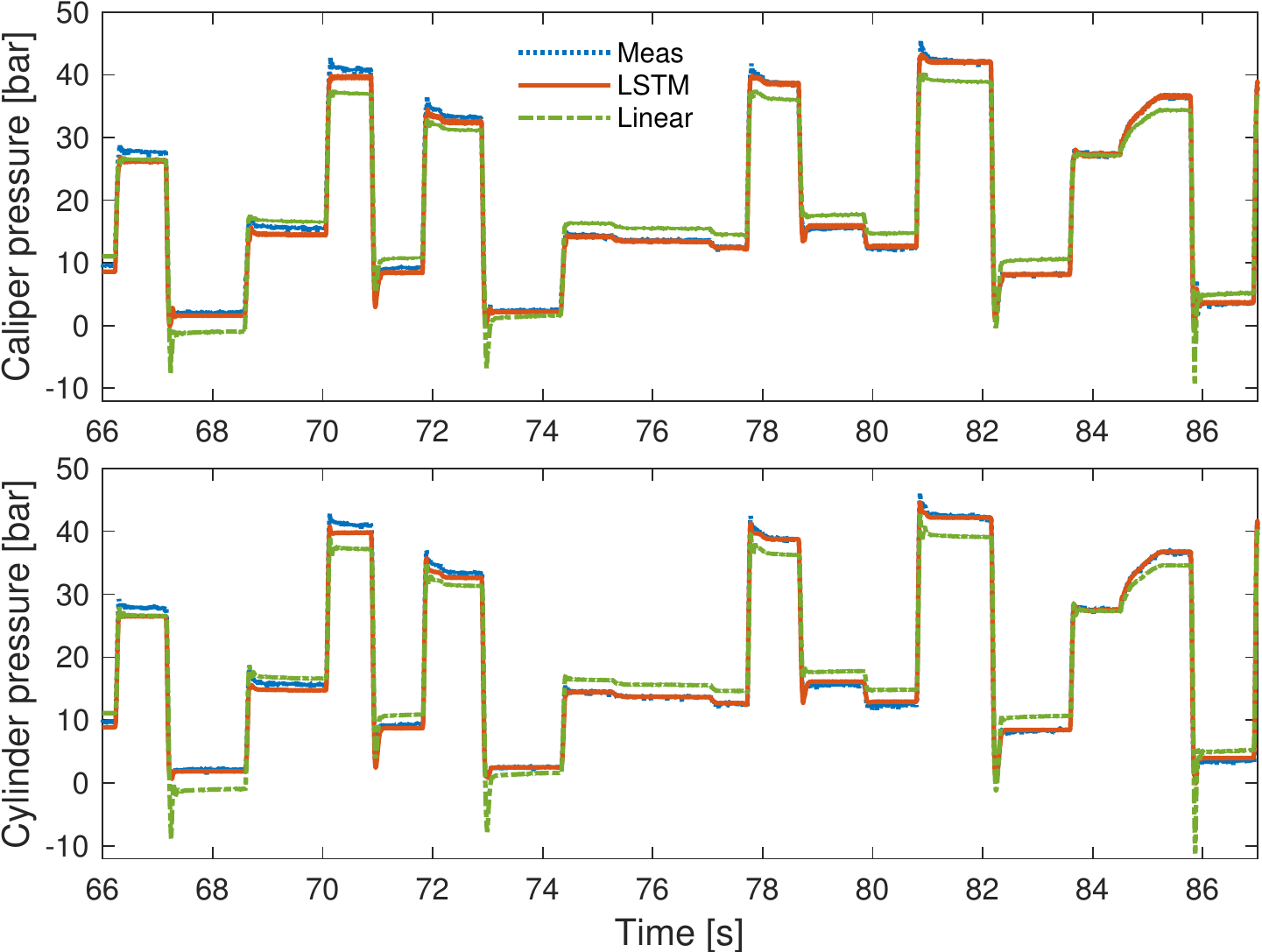} \\
	\caption{Performances testing: open-loop prediction of the trained deep LSTM model (red solid line) compared to a linear model (green dash-dotted line) and the ground truth (blue dotted line).}
	\label{fig:example:Y_test_LSTM_vs_linear}
\end{figure}
The effectiveness of the LSTM model is, indeed, manifest: the linear model, despite an overall agreement with both pressures, features larger steady-state errors and undershoots, which even cause unrealistic negative pressure values. In order to quantify the LSTM performance the FIT index is used see Appendix~\ref{appendix:training}, resulting in $\textrm{FIT} = 88.6 \%$, indicating good modeling performances.

Finally, in order to appreciate the network capability of precisely modeling also the dynamics which characterize the cylinder and caliper pressures, Figure \ref{fig:example:Y_test_dynamic} shows a close up of the test data during a highly dynamic transient: the learned model is capable of correctly describing the system cause-effect relationship, i.e. the cylinder pressure increases before the caliper one; moreover it well describes the respective dynamics, in particular the less damped behavior of the caliper pressure compared to the master cylinder one.
\begin{figure}[t]
	\centering
	\includegraphics[width=\columnwidth]{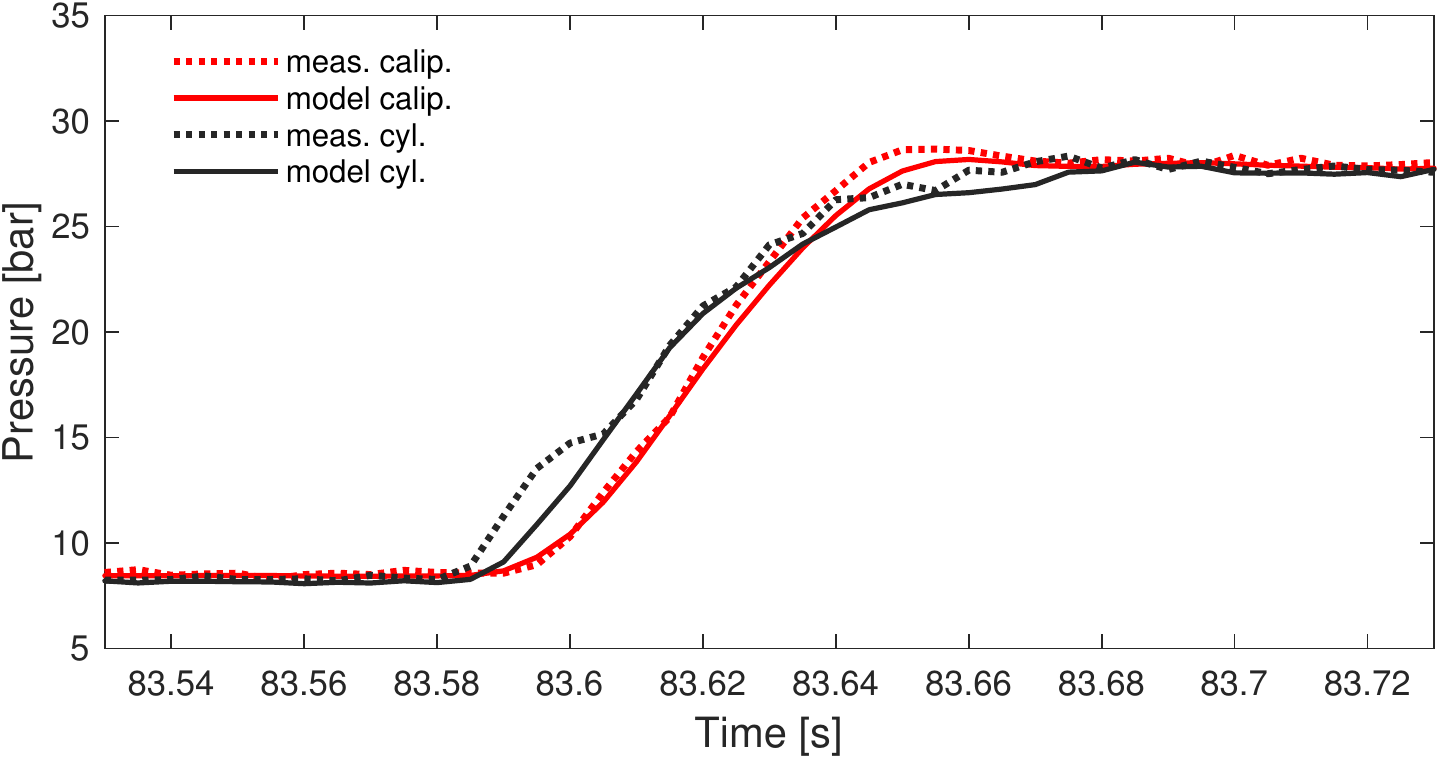} \\
	\caption{Performances testing: LSTM modeling detail during a highly dynamic transient.}
	\label{fig:example:Y_test_dynamic}
\end{figure}

\section{Conclusions}
In this paper, the use of deep Long Short-Term Memory networks (LSTMs) for the identification of nonlinear dynamical systems has been investigated.
Sufficient conditions guaranteeing their Incremental Input-to-State Stability ($\delta$ISS) were devised.
These conditions consist of a number of nonlinear inequalities on the weights of the network, and can be enforced during the training procedure.
 
The proposed approach was tested to identify a brake-by-wire system from data collected experimentally from the apparatus, yielding more than satisfactory results.
Future research directions will concern the use of the identified network for the synthesis of a computationally efficient control system for the benchmark system under analysis, e.g., the Internal Model Control strategy proposed in \cite{bonassi2022imc}.

\appendices
\section{Proofs} \label{appendix:proofs}
\subsection{Proof of Proposition \ref{prop:gates_bounds}} \label{proof:gates_bounds}

First, let us notice that, for any layer $l \in \{ 1, ..., L \}$, it holds that $h_k^{(l)}$ is unity-bounded, i.e., $\| h_k^{(l)} \|_\infty \leq 1$, since it is defined as the element-wise product between two subunitary vectors, see \eqref{eq:lstm:model:statespace}.
Owing to Assumption \ref{ass:input}, and in view of the definition of the input vector \eqref{eq:lstm:model:input}, it holds that $\| u_{k}^{(l)} \|_\infty \leq 1$. 

The bounds \eqref{eq:lstm:gates_bounds} can be therefore obtained by applying layer-wise the arguments devised in \cite[Section~2.2]{terzi2021lstm}. 
\hfill\QED

\subsection{Proof of Proposition \ref{prop:invariant_set}} \label{proof:invariant_set}
Since, for each layer $l \in \{ 1, ..., L \}$, the input vector $u_k^{(l)}$ is subunitary, by applying the results obtained in \cite[Section~2.2]{terzi2021lstm} layer-wise, then~\eqref{eq:lstm:model:statespace} admits the invariant set $\mathcal{X}^{(l)}$ reported in \eqref{eq:lstm:invset}. 
Thus, \eqref{eq:lstm:invset:full} is an invariant set of the deep LSTM model.
\hfill\QED

\subsection{Proof of Theorem \ref{theorem:deltaiss}} \label{proof:deltaiss}
Consider a generic layer $l \in \{ 1, ..., L \}$. 
Owing to the unity-boundedness of its input vector, the satisfaction of conditions \eqref{eq:stability:conditions} implies that the layer is $\delta$ISS by \cite[Proposition~2]{terzi2021lstm}.
That is, considering the pair of initial states $x_{a, 0} \in \mathcal{X}$ and $x_{b, 0} \in \mathcal{X}$, and the pair of input sequences $u_{a, 0:k} \in \mathcal{U}_{0:k}$ and $u_{b, 0:k} \in \mathcal{U}_{0:k}$, and letting $x_{\alpha, k}$ be short for $x_k(x_{\alpha, 0}, u_{\alpha, 0:k})$, it holds that
\begin{equation}
\begin{aligned}
	\begin{bmatrix}
		\| c_{a, k+1}^{(l)} - c_{b, k+1}^{(l)} \|_2 \\ 
		\| h_{a, k+1}^{(l)} - h_{b, k+1}^{(l)} \|_2   
	\end{bmatrix} \leq& \mathfrak{A}^{(l)} \begin{bmatrix}
		\| c_{a, k}^{(l)} - c_{b, k}^{(l)} \|_2 \\ 
		\| h_{a, k}^{(l)} - h_{b, k}^{(l)} \|_2   
	\end{bmatrix} \\
	&+ \mathfrak{B}^{(l)} \| u_{a, k}^{(l)} - u_{b, k}^{(l)} \|_2,
\end{aligned}
\end{equation}
where $\mathfrak{A}^{(l)}$ a Schur stable 2-by-2 matrix, see \cite{terzi2021lstm}, defined as 
\begin{equation}
	\mathfrak{A}^{(l)} = \begin{bmatrix}
		\bar{\sigma}_f^{(l)} & \bar{\alpha}^{(l)} \\
		\bar{\sigma}_z^{(l)} \bar{\sigma}_f^{(l)} & \bar{\sigma}_z^{(l)} \bar{\alpha}^{(l)} + \frac{1}{4} \bar{h} \| U_z^{(l)} \|_2
	\end{bmatrix}.
\end{equation}
Recalling the definition of $u_{\alpha, k}^{(l)}$ given in \eqref{eq:lstm:model:input}, one obtains that
\begin{subequations} \label{eq:proof:deltaiss:ineq_k}
\begin{equation}
\begin{aligned}
	\begin{bmatrix}
		\| c_{a, k+1}^{(1)} - c_{b, k+1}^{(1)} \|_2 \\ 
		\| h_{a, k+1}^{(1)} - h_{b, k+1}^{(1)} \|_2  \\
		\vdots \\
		\| c_{a, k+1}^{(L)} - c_{b, k+1}^{(L)} \|_2 \\ 
		\| h_{a, k+1}^{(L)} - h_{b, k+1}^{(L)} \|_2  \\
	\end{bmatrix} \leq& \mathfrak{A}
	\begin{bmatrix}
		\| c_{a, k+1}^{(1)} - c_{b, k+1}^{(1)} \|_2 \\ 
		\| h_{a, k+1}^{(1)} - h_{b, k+1}^{(1)} \|_2  \\
		\vdots \\
		\| c_{a, k+1}^{(L)} - c_{b, k+1}^{(L)} \|_2 \\ 
		\| h_{a, k+1}^{(L)} - h_{b, k+1}^{(L)} \|_2  \\
	\end{bmatrix}  \\
	&+ \mathfrak{B} \| u_{a, k} - u_{b, k} \|_2,
\end{aligned}
\end{equation}
where
\begin{equation}
\scalemath{0.8}{
	\mathfrak{A} = \begin{bmatrix}
		\mathfrak{A}^{(1)} & 0_{2, 2} & ... & 0_{2, 2} \\
		\mathfrak{B}^{(2)} \tilde{\mathfrak{A}}^{(1)} & \mathfrak{A}^{(2)} & ... & 0_{2, 2} \\
		\vdots && \ddots & \vdots \\
		\mathfrak{B}^{(L)}  \big( \prod_{l= L - 1}^{2}  \tilde{\mathfrak{A}}^{(l)} \big) {\mathfrak{A}}^{(1)}  & \mathfrak{B}^{(L)} \big( \prod_{l= L - 1}^{3}   \tilde{\mathfrak{A}}^{(l)} \big) {\mathfrak{A}}^{(2)}  & ... & \mathfrak{A}^{(L)}
	\end{bmatrix},}
\end{equation}
and 
\begin{equation}
\scalemath{0.8}{
	\mathfrak{B} = \begin{bmatrix}
		 \mathfrak{B}^{(1)} \\
		  \mathfrak{B}^{(2)} \tilde{\mathfrak{B}}^{(1)} \\
		  \vdots \\
		  \mathfrak{B}^{(L)} \big( \prod_{l=L-1}^{1} \tilde{\mathfrak{B}}^{(l)} \big)
	\end{bmatrix},}
\end{equation}
\end{subequations}
while the matrices $\tilde{\mathfrak{A}}^{(l)}$ and $\tilde{\mathfrak{B}}^{(l)}$ are defined as $\tilde{\mathfrak{A}}^{(l)} = [ 0, 1 ] \mathfrak{A}^{(l)}$ and $\tilde{\mathfrak{B}}^{(l)} = [ 0, 1 ] \mathfrak{B}^{(l)}$.
By iterating \eqref{eq:proof:deltaiss:ineq_k}, one gets
\begin{equation}
\begin{aligned}
	\begin{bmatrix}
		\| c_{a, k}^{(1)} - c_{b, k}^{(1)} \|_2 \\ 
		\| h_{a, k}^{(1)} - h_{b, k}^{(1)} \|_2  \\
		\vdots \\
		\| c_{a, k}^{(L)} - c_{b, k}^{(L)} \|_2 \\ 
		\| h_{a, k}^{(L)} - h_{b, k}^{(L)} \|_2  \\
	\end{bmatrix} \leq& \mathfrak{A}^k \begin{bmatrix}
		\| c_{a, 0}^{(1)} - c_{b, 0}^{(1)} \|_2 \\ 
		\| h_{a, 0}^{(1)} - h_{b, 0}^{(1)} \|_2  \\
		\vdots \\
		\| c_{a, 0}^{(L)} - c_{b, 0}^{(L)} \|_2 \\ 
		\| h_{a, 0}^{(L)} - h_{b, 0}^{(L)} \|_2  \\
	\end{bmatrix}  \\
	& + \sum_{\tau = 0}^{k-1} \mathfrak{A}^{k - \tau -1} \mathfrak{B} \| u_{a, k} - u_{b, k} \|_2
\end{aligned}
\end{equation}
Let us point out that
\begin{equation}
\scalemath{0.9}{
	\left\|
	\begin{bmatrix}
		\| c_{a, k}^{(1)} - c_{b, k}^{(1)} \|_2 \\ 
		\| h_{a, k}^{(1)} - h_{b, k}^{(1)} \|_2  \\
		\vdots \\
		\| c_{a, k}^{(L)} - c_{b, k}^{(L)} \|_2 \\ 
		\| h_{a, k}^{(L)} - h_{b, k}^{(L)} \|_2  \\
	\end{bmatrix} \right\|_2 = \left\| \begin{bmatrix}
		\| x_{a, k}^{(1)} - x_{b, k}^{(1)} \|_2 \\
		\vdots \\
		\| x_{a, k}^{(L)} - x_{b, k}^{(L)} \|_2  \\
	\end{bmatrix} \right\|_2 = \| x_{a, k} - x_{b, k} \|_2.}
\end{equation}
Moreover, the matrix $\mathfrak{A}$ is Schur stable, since it is a block-diagonal matrix with Schur stable blocks on the main diagonal. 
Thus, there exist $\lambda \in (0, 1)$ and $\mu > 0$ such that
\begin{equation}
\scalemath{0.95}{
\begin{aligned}
	\| x_{a, k} - x_{b, k} \|_2 \leq& \mu \lambda^k \| x_{a, 0} - x_{b, 0} \|_2 \\
	&+ \| (I_{2L, 2L} - \mathfrak{A})^{-1} \mathfrak{B} \|_2 \| u_{a, 0:k} - u_{b, 0:k} \|_{2, \infty},
\end{aligned}}
\end{equation}
i.e., the deep LSTM is $\delta$ISS.
\hfill \QED

\section{Training procedure} \label{appendix:training}
In this Appendix, the algorithm here adopted to train the deep LSTM model \eqref{eq:lstm:compact}, based on the so-called Truncated BackPropagation Through Time (TBPTT) \cite{bianchi2017recurrent}, is briefly described. For more details, the interested reader is addressed to \cite{bonassi2023reconciling}.
This procedure represents an intuitive way to address the training problem, i.e., the problem of finding the weights $\Phi^\star$ that minimize the free-run simulation error over the training data.
In summary, TBPTT consists of extracting $N_s$ partially-overlapping subsequences from the available input-output sequences collected during the experiment campaign. 
These subsequences, denoted as $\big(u_{0:T_s}^{\{ i\}}, y_{0:T_s}^{\{ i\}} \big)$, are indexed by $i \in \mathcal{I} = \{1, ..., N_s \}$, and have fixed length $T_s$, which is a design choice.
The set of subsequences $\mathcal{I}$ is split in a training set $\mathcal{I}_{tr}$, used to effectively train the model, and a validation set $\mathcal{I}_{val}$, which is instead used to quantify the modeling performances of the network during its iterative training procedure.

The training procedure is carried out in an iterative and batched fashion.
At every iteration (also called \emph{epoch}), the set of training subsequences, i.e. $\mathcal{I}_{tr}$, is randomly partition in $B$ subsets $\mathcal{I}_{tr}^{\{ b\}}$, with $b \in \{ 1, ..., B \}$, called batches.
Then, for each batch $b$, the loss function
\begin{equation} \label{eq:training:loss}
	\begin{aligned}
		\mathcal{L}(\Phi) &= \textrm{MSE} \big(\mathcal{I}_{tr}^{\{b\}}; \Phi \big) + \rho(\nu(\Phi)),
	\end{aligned}
\end{equation}
is evaluated, its gradient with respect to the weights, i.e. $\nabla_{\Phi}(\mathcal{L}(\Phi))$, is symbolically computed, and a gradient descend step is performed.
In its simplest form, this boils down to updating the weights as $\Phi \gets \Phi - \gamma \nabla_{\Phi}(\mathcal{L}(\Phi))$,
where $\gamma$ is a sufficiently small learning rate, see \cite{goodfellow2016deep}. 
More advanced techniques, based on gradient's momentum, can also be employed to speed up the convergence of this procedure.

Note that the loss function \eqref{eq:training:loss} consists of two terms.
The former term is the Mean Square Error (MSE) scored by the deep LSTM's free run simulation on the batch $\mathcal{I}_{tr}^{\{b\}}$, i.e.,
\begin{equation}
	\textrm{MSE}\big(\mathcal{I}_{tr}^{\{b\}}; \Phi \big) =  \sum_{i \in \mathcal{I}_{tr}^{\{b\}}} \sum_{k = \tau_w + 1}^{T_s} \frac{\| y_{k}(x_0, u_{0:k}^{\{ i \}}) - y^{\{ i \}}_{k} \|_2^2}{\lvert \mathcal{I}_{tr}^{\{b\}} \lvert (T_s - \tau_w)},
\end{equation}  
where $\tau_w \geq 0$ is the washout period, that allows to accomodate for the initial transient due to the random initial state $x_0$, see \cite{bianchi2017recurrent}.

The second term in \eqref{eq:training:loss} is a suitably designed regularization term that enforces the satisfaction of the $\delta$ISS sufficient conditions \eqref{eq:stability:conditions}, see \cite{bonassi2022survey, bonassi2023reconciling}.
An example of such regularization function is a piecewise-linear function, i.e.,
\begin{equation} \label{eq:appendix:regularization}
	\rho(\nu) = \frac{1}{n_{\nu}} \sum_{i=1}^{n_{\nu}} \big[ \bar{\pi} \max([\nu]_i + \varepsilon_{\nu}, 0) + \ubar{\pi} \min([\nu]_i + \varepsilon_{\nu}, 0)\big],
\end{equation}
where $n_{\nu}$ is the number of elements of $\nu(\Phi)$, i.e. the number of inequalities, $\varepsilon_\nu > 0$ is the constraint clearance, while  $\bar{\pi} \gg \ubar{\pi} > 0$ represent the cost coefficients.
More advanced regularization functions, such as the generalized piecewise function \cite{bonassi2023reconciling}, can also be adopted. 
This iterative training procedure is eventually halted when the performances on validation set, i.e. the validation metrics $\textrm{MSE}(\mathcal{I}_{val}; \Phi)$, stop improving, yielding the trained weights, denoted by $\Phi^\star$.

Finally, the performances of the trained model need to be objectively quantified on the independent test set.
To this end, the $\textrm{FIT}$ performance index $[\%]$ can be adopted \cite{bonassi2022survey, bonassi2023reconciling}, which is defined as
\begin{equation}
    \textrm{FIT} = 100 \left( 1 - \sum_{k=\tau_w+1}^{T_{te}} \frac{\| y_k - y_{k}^{\{te\}
} \|_2}{\| y_{k}^{\{te\}} - \bar{y}^{\{te\}} \|_2} \right),
\end{equation}
where $y^{\{te\}}_{0:T_{te}}$ denotes the output sequence of the test dataset, $y_k = y_{k}(x_0, u_{0:k}^{\{te\}})$ denotes the output of the LSTM model \eqref{eq:lstm:compact}, and $\bar{y}^{\{te\}}$ indicates the average value of $y_{0:T_{te}}^{\{te\}}$.
Note that a $\textrm{FIT}$ index of $100 \%$ indicates ideal modeling performances.

\section{Learning a black-box model of the BBW apparatus} \label{appendix:example}

\begin{figure}[t]
	\centering
	\includegraphics[width=\columnwidth]{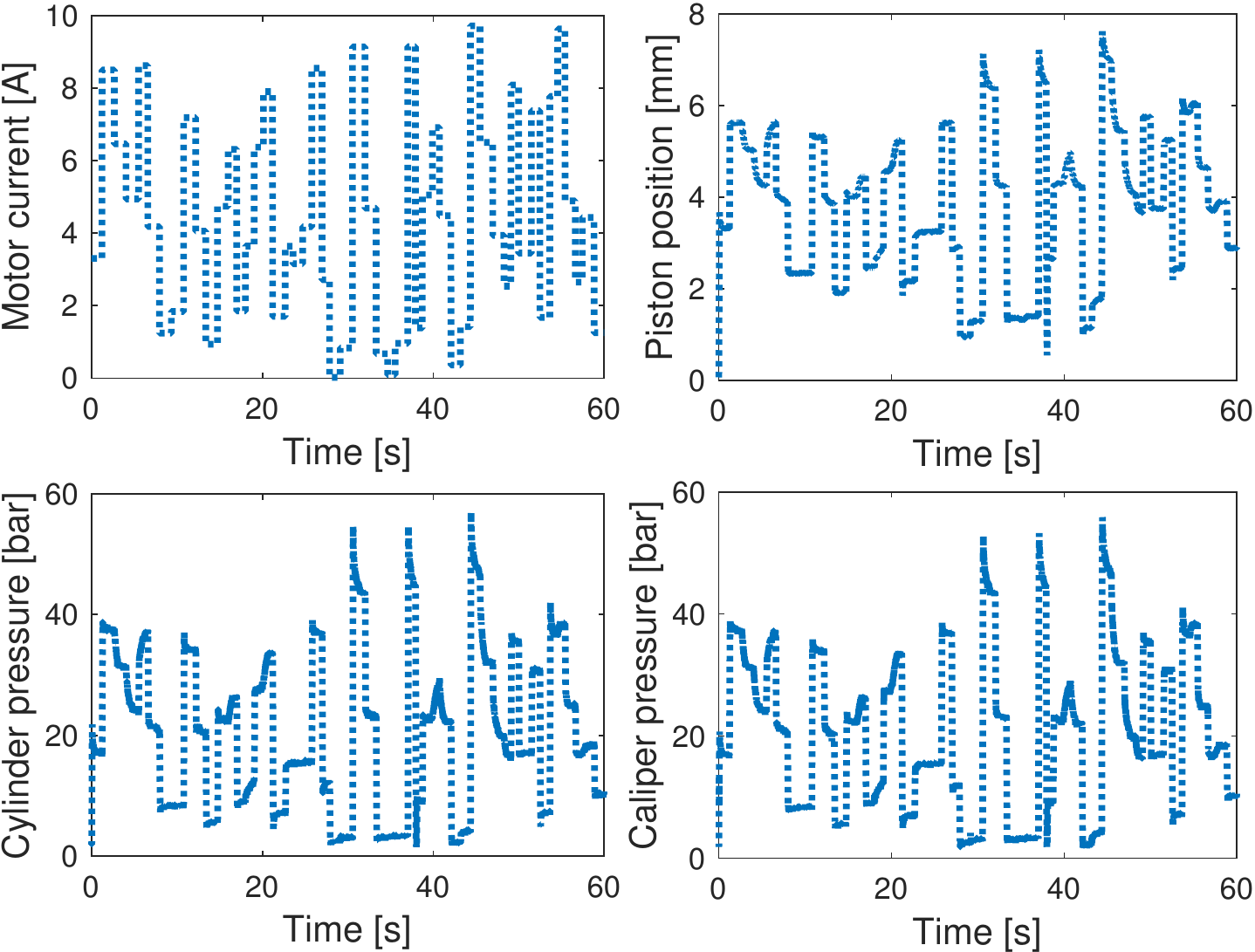}
	\caption{Portion of experiment $1$: applied current reference (top left), measured piston position (top right), measured pressure at the master cylinder (bottom left), and measured caliper pressure (bottom right).}
	\label{fig:example:data}
\end{figure}

As mentioned in Section IV, 6 open loop experiments have been conducted on the BBW apparatus: an example of the collected sequences is shown in Figure~\ref{fig:example:data}, where the first $60$ seconds of experiment $1$ are reported. 

After these sequences are collected, they are normalized \cite{bonassi2023reconciling}, and are split as follows: experiment $1$ has been partitioned into $40\%$ validation and $60 \%$ testing, whereas the remaining experiments into $25 \%$ validation and $75\%$ training.
Then, according to the TBPTT algorithm discussed in Appendix \ref{appendix:training}, $N_{tr}= 3500$ training subsequences and $N_{val} = 1000$ validation subsequences, of length $T_s = 200$, have been extracted from these datasets.
The former subsequences are collected in the training set $\mathcal{I}_{tr}$, the latter in the validation set $\mathcal{I}_{val}$.
Note that the test sequence, in contrast, is not partitioned into subsequences, since it is used to assess the modeling performance of the model at the end of the training procedure over a sufficiently long open-loop simulation.

The loss function is then designed as \eqref{eq:training:loss}, where the $\delta$ISS-enforcing term $\rho(\nu(\Phi))$ is defined as in \eqref{eq:appendix:regularization}, with $\bar{\pi} = 2 \cdot 10^{-4}$, $\ubar{\pi} = 10^{-6}$, and $\varepsilon_\nu = 0.02$. 

The training has been carried out using the RMSProp optimization algorithm \cite{goodfellow2016deep}, and batches of $50$ training subsequences, which implies that $B = 70$ batches need to be processed at each training epoch.
In Figure \ref{fig:example:training}, the evolution of the training loss function \eqref{eq:training:loss} and of the $\text{MSE}$ metric on the the validation dataset, i.e. $\text{MSE}(\mathcal{I}_{val}; \Phi)$, have been reported. 
The overall training procedure has been carried out for $694$ epochs, and has been halted when the validation metrics $\text{MSE}(\mathcal{I}_{val}; \Phi)$ stopped to improve, yielding the trained weights $\Phi^\star$.
The resulting deep LSTM model satisfies the $\delta$ISS sufficient conditions \eqref{eq:stability:conditions}, with $\max_i [\nu(\Phi^\star)]_i < -0.09 $.

\begin{figure}[t]
	\centering
	\includegraphics[width=\columnwidth]{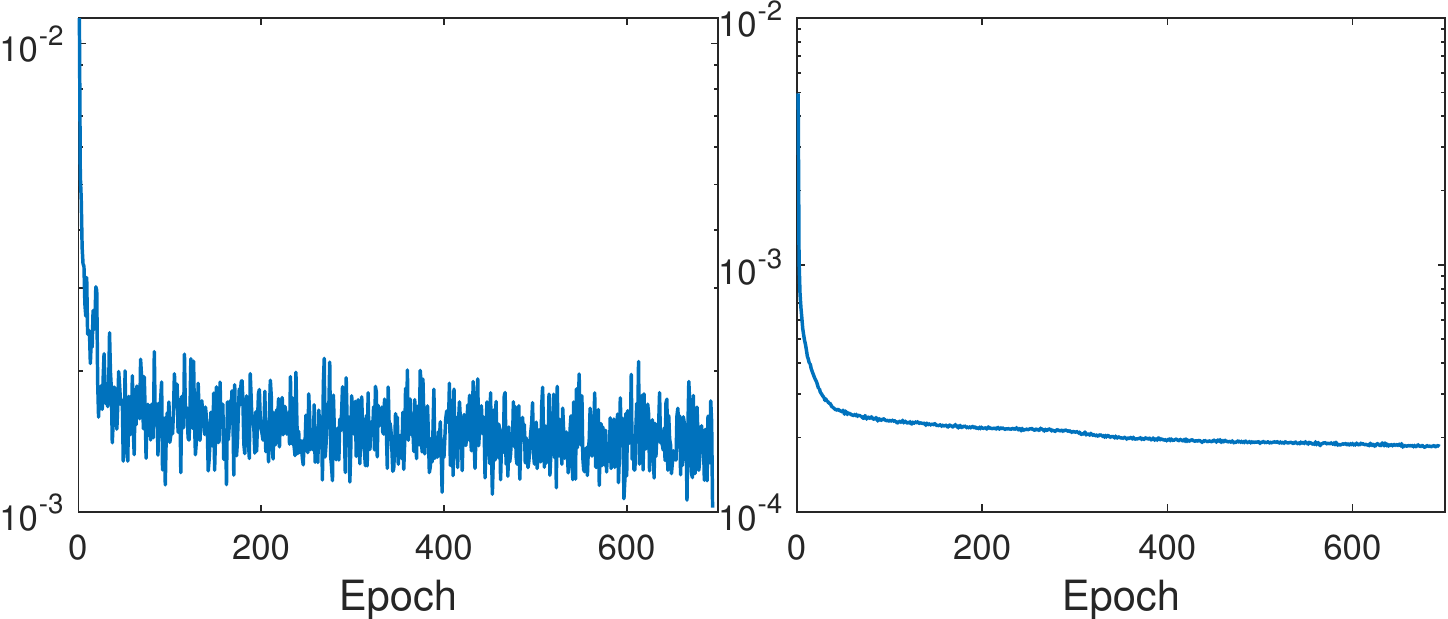}
	\caption{Evolution of the loss function (left) and of the $\text{MSE}$ over the validation dataset (right) throughout the training procedure.}
	\label{fig:example:training}
\end{figure}

\section*{Acknowledgments}
This project has received funding from the European Union’s Horizon 2020 research and innovation programme under the Marie Skłodowska-Curie grant No. 953348.

This work has been partially supported by the Italian Ministry for Research in the framework of the 2017 Program for Research Projects of National Interest (PRIN), Grant No.~2017YKXYXJ.

\bibliographystyle{IEEEtran}
\bibliography{Bibliografia}

\begin{thebibliography}{10}
\providecommand{\url}[1]{#1}
\csname url@samestyle\endcsname
\providecommand{\newblock}{\relax}
\providecommand{\bibinfo}[2]{#2}
\providecommand{\BIBentrySTDinterwordspacing}{\spaceskip=0pt\relax}
\providecommand{\BIBentryALTinterwordstretchfactor}{4}
\providecommand{\BIBentryALTinterwordspacing}{\spaceskip=\fontdimen2\font plus
\BIBentryALTinterwordstretchfactor\fontdimen3\font minus \fontdimen4\font\relax}
\providecommand{\BIBforeignlanguage}[2]{{%
\expandafter\ifx\csname l@#1\endcsname\relax
\typeout{** WARNING: IEEEtran.bst: No hyphenation pattern has been}%
\typeout{** loaded for the language `#1'. Using the pattern for}%
\typeout{** the default language instead.}%
\else
\language=\csname l@#1\endcsname
\fi
#2}}
\providecommand{\BIBdecl}{\relax}
\BIBdecl

\bibitem{hunt1992neural}
K.~J. Hunt \emph{et~al.}, ``Neural networks for control systems -- a survey,'' \emph{Automatica}, vol.~28, no.~6, pp. 1083--1112, 1992.

\bibitem{goodfellow2016deep}
Y.~Bengio, I.~Goodfellow, and A.~Courville, \emph{Deep learning}.\hskip 1em plus 0.5em minus 0.4em\relax MIT press Massachusetts, USA, 2017, vol.~1.

\bibitem{bianchi2017recurrent}
F.~M. Bianchi, E.~Maiorino, M.~C. Kampffmeyer, A.~Rizzi, and R.~Jenssen, \emph{Recurrent neural networks for short-term load forecasting: an overview and comparative analysis}.\hskip 1em plus 0.5em minus 0.4em\relax Springer, 2017.

\bibitem{hochreiter1997long}
S.~Hochreiter and J.~Schmidhuber, ``{Long Short-Term Memory},'' \emph{Neural Computation}, vol.~9, no.~8, pp. 1735--1780, 11 1997.

\bibitem{cho2014learning}
K.~Cho \emph{et~al.}, ``Learning phrase representations using {RNN} encoder-decoder for statistical machine translation,'' in \emph{Conference on Empirical Methods in Natural Language Processing (EMNLP 2014)}, 2014.

\bibitem{pascanu2013difficulty}
R.~Pascanu, T.~Mikolov, and Y.~Bengio, ``On the difficulty of training recurrent neural networks,'' in \emph{International conference on machine learning}.\hskip 1em plus 0.5em minus 0.4em\relax PMLR, 2013, pp. 1310--1318.

\bibitem{mohajerin2019multistep}
N.~Mohajerin and S.~L. Waslander, ``Multistep prediction of dynamic systems with recurrent neural networks,'' \emph{IEEE transactions on neural networks and learning systems}, vol.~30, no.~11, pp. 3370--3383, 2019.

\bibitem{xu2022review}
J.~Xu, M.~Kovatsch, D.~Mattern, F.~Mazza, M.~Harasic, A.~Paschke, and S.~Lucia, ``A review on {AI} for smart manufacturing: Deep learning challenges and solutions,'' \emph{Applied Sciences}, vol.~12, no.~16, p. 8239, 2022.

\bibitem{bonassi2022survey}
F.~Bonassi, M.~Farina, J.~Xie, and R.~Scattolini, ``{O}n {R}ecurrent {N}eural {N}etworks for learning-based control: recent results and ideas for future developments,'' \emph{Journal of Process Control}, vol. 114, pp. 92--104, 2022.

\bibitem{schoukens2019nonlinear}
J.~Schoukens and L.~Ljung, ``Nonlinear system identification: A user-oriented road map,'' \emph{IEEE Control Systems Magazine}, vol.~39, no.~6, pp. 28--99, 2019.

\bibitem{rawlings2017model}
J.~B. Rawlings, D.~Q. Mayne, and M.~Diehl, \emph{Model predictive control: theory, computation, and design}.\hskip 1em plus 0.5em minus 0.4em\relax Nob Hill Publishing Madison, WI, 2017, vol.~2.

\bibitem{wu2019machine}
Z.~Wu, A.~Tran, D.~Rincon, and P.~D. Christofides, ``Machine learning-based predictive control of nonlinear processes. part i: Theory,'' \emph{AIChE Journal}, vol.~65, no.~11, 2019.

\bibitem{levin1996control}
A.~Levin and K.~Narendra, ``Control of nonlinear dynamical systems using neural networks: controllability and stabilization,'' \emph{IEEE Transactions on Neural Networks}, vol.~4, no.~2, pp. 192--206, 1993.

\bibitem{terzi2021lstm}
E.~Terzi, F.~Bonassi, M.~Farina, and R.~Scattolini, ``Learning model predictive control with long short-term memory networks,'' \emph{International Journal of Robust and Nonlinear Control}, vol.~31, no.~18, pp. 8877--8896, 2021.

\bibitem{lanzetti2019recurrent}
N.~Lanzetti, Y.~Z. Lian, A.~Cortinovis, L.~Dominguez, M.~Mercang{\"o}z, and C.~Jones, ``Recurrent neural network based {MPC} for process industries,'' in \emph{2019 18th European Control Conference (ECC)}.\hskip 1em plus 0.5em minus 0.4em\relax IEEE, 2019, pp. 1005--1010.

\bibitem{nagy2007model}
Z.~K. Nagy, ``Model based control of a yeast fermentation bioreactor using optimally designed artificial neural networks,'' \emph{Chemical engineering journal}, vol. 127, no. 1-3, pp. 95--109, 2007.

\bibitem{wong2018recurrent}
W.~C. Wong, E.~Chee, J.~Li, and X.~Wang, ``Recurrent neural network-based model predictive control for continuous pharmaceutical manufacturing,'' \emph{Mathematics}, vol.~6, no.~11, p. 242, 2018.

\bibitem{bayer2013discrete}
F.~Bayer, M.~B{\"u}rger, and F.~Allg{\"o}wer, ``Discrete-time incremental {ISS}: A framework for robust {NMPC},'' in \emph{2013 European Control Conference (ECC)}.\hskip 1em plus 0.5em minus 0.4em\relax IEEE, 2013, pp. 2068--2073.

\bibitem{bonassi2022imc}
F.~Bonassi and R.~Scattolini, ``Recurrent neural network-based {I}nternal {M}odel {C}ontrol of unknown nonlinear stable systems,'' \emph{European Journal of Control}, p. 100632, 2022.

\bibitem{bonassi2021nnarx}
F.~Bonassi, M.~Farina, and R.~Scattolini, ``Stability of discrete-time feed-forward neural networks in {NARX} configuration,'' in \emph{19th IFAC Symposium on System Identification (SYSID 2021)}, 2021.

\bibitem{bonassi2022offset}
F.~Bonassi, M.~Farina, J.~Xie, and R.~Scattolini, ``{A}n {O}ffset-{F}ree {N}onlinear {MPC} scheme for systems learned by {N}eural {NARX} models,'' in \emph{61st IEEE Conference on Decision and Control (CDC)}, 2022.

\bibitem{bonassi2020lstm}
F.~Bonassi, E.~Terzi, M.~Farina, and R.~Scattolini, ``{LSTM} neural networks: {I}nput to state stability and probabilistic safety verification,'' in \emph{Learning for Dynamics and Control}.\hskip 1em plus 0.5em minus 0.4em\relax PMLR, 2020, pp. 85--94.

\bibitem{bonassi2021stability}
F.~Bonassi, M.~Farina, and R.~Scattolini, ``On the stability properties of gated recurrent units neural networks,'' \emph{Systems \& Control Letters}, vol. 157, p. 105049, 2021.

\bibitem{bonassi2021nonlinear}
F.~Bonassi, C.~F. Oliveira~da Silva, and R.~Scattolini, ``Nonlinear {MPC} for {O}ffset-{F}ree {T}racking of systems learned by {GRU} {N}eural {N}etworks,'' in \emph{3rd IFAC Conference on Modelling, Identification and Control of Nonlinear Systems (MICNON 2021)}, 2021.

\bibitem{angeli2002lyapunov}
D.~Angeli, ``A lyapunov approach to incremental stability properties,'' \emph{IEEE Transactions on Automatic Control}, vol.~47, no.~3, pp. 410--421, 2002.

\bibitem{todeschini2014adaptive}
F.~Todeschini, M.~Corno, G.~Panzani, S.~Fiorenti, and S.~M. Savaresi, ``Adaptive cascade control of a brake-by-wire actuator for sport motorcycles,'' \emph{IEEE/ASME Transactions on Mechatronics}, vol.~20, no.~3, pp. 1310--1319, 2014.

\bibitem{corno2015BBWmodeling}
M.~Corno, F.~Todeschini, G.~Panzani, and S.~M. Savaresi, ``Modeling and parameter identification of a brake-by-wire actuator for racing motorcycles,'' in \emph{Control-Oriented Modelling and Identification: Theory and Practice}.\hskip 1em plus 0.5em minus 0.4em\relax The Institution of Engineering and Technology, 2015, p. 329 – 363.

\bibitem{bonassi2023reconciling}
F.~Bonassi, ``Reconciling deep learning and control theory: recurrent neural networks for model-based control design,'' 2023, {D}octoral {D}issertation.

\end{thebibliography}

\end{document}